\documentclass[aps,prl,reprint,longbibliography]{revtex4-2}

\usepackage{amsfonts,amscd,mathrsfs,amsmath,amsthm,amssymb}
\usepackage{mathtools}
\usepackage{xparse}
\usepackage{graphicx}
\usepackage{float}
\usepackage{booktabs}

\usepackage[dvipsnames]{xcolor}

\PassOptionsToPackage{hyphens}{url}\usepackage{hyperref}

\hypersetup{
    colorlinks=true,
    linkcolor=magenta,
    filecolor=black,      
    citecolor=magenta,
}

\usepackage[capitalise]{cleveref}
\usepackage{physics}
\usepackage{bm}
\usepackage{bbm}
\usepackage{caption}

\theoremstyle{definition}

\newtheorem{theorem}{Theorem}

\newtheorem{lemma}{Lemma}

\newcommand\numberthis{\addtocounter{equation}{1}\tag{\theequation}}

    \newcommand{\bfit}[1]{\textcolor{black}{\textit{\textbf{#1}}}}

    \NewDocumentCommand\explainequals{m}{\overset{\textit{#1}}{=}}

    \newcommand{\Z}{\mathbb{Z}}	
    \newcommand{\C}{\mathbb{C}} 
        \newcommand{\Q}{\mathbb{Q}} 

    \newcommand{\indicator}{\mathbbm{1} }
    
    \DeclareMathAlphabet{\mathsfit}{T1}{\sfdefault}{\mddefault}{\sldefault}
    \SetMathAlphabet{\mathsfit}{bold}{T1}{\sfdefault}{\bfdefault}{\sldefault}

    \newcommand{\X}{\mathsf{X}}
    \newcommand{\Y}{\mathsf{Y}}
    \renewcommand{\Z}{\mathsf{Z}}
        \renewcommand{\S}{\mathsf{S}}
    \renewcommand{\H}{\mathsf{H}}
        \newcommand{\F}{\mathsf{F}}

    \newcommand{\Ph}{\mathsf{Ph}}
    \renewcommand{\P}{\mathsf{P}}
    \newcommand{\G}{\mathsf{G}}

    \newcommand{\logzero}{\ket{\overline{0}}}
    \newcommand{\logone}{\ket{\overline{1}}}

     \newcommand{\logicalket}[1]{\ket*{\overline{#1}}}
    \newcommand{\logicalbra}[1]{\bra*{\overline{#1}}}

    \newcommand{\SL}{\mathrm{SL}}
    \newcommand{\SU}{\mathrm{SU}}
    \newcommand{\U}{\mathrm{U}}
    \newcommand{\SO}{\mathrm{SO}}
    \newcommand{\D}{\mathscr{D}}

        \newcommand{\tet}{2\mathrm{T}}
        \newcommand{\oct}{2\mathrm{O}}
        \newcommand{\ico}{2\mathrm{I}}

    \newcommand{\smallminus}{\text{-}}

\begin{document}

\title{A Family of Quantum Codes with Exotic Transversal Gates}

\thanks{ These authors contributed equally to this work.}

\author{Eric Kubischta}

\email{erickub@umd.edu} 

\author{Ian Teixeira}

\email{igt@umd.edu}

\affiliation{Joint Center for Quantum Information and Computer Science,
NIST/University of Maryland, College Park, Maryland 20742 USA}

\begin{abstract}
    Recently an algorithm has been constructed that shows the binary icosahedral group $\ico$ together with a $T$-like gate forms the most efficient single-qubit universal gate set. To carry out the algorithm fault tolerantly requires a code that implements $\ico$ transversally. However, no such code has ever been demonstrated in the literature. We fill this void by constructing a family of distance $d = 3$ codes that all implement $\ico$ transversally. A surprising feature of this family is that the codes can be deduced entirely from symmetry considerations that only $\ico$ affords. 
\end{abstract}

\maketitle

\section{Introduction}

Let $((n,K,d))$ denote an $n$-qubit quantum error-correcting code  with a codespace of dimension $K$ and distance $d$. The Eastin-Knill theorem \cite{eastinKnill} shows that when a code is non-trivial ($ d \geq 2 $), the logical operations in $ \SU(K) $ which can be implemented transversally are always a \textit{finite} subgroup $ \G \subset \SU(K) $. A logical gate $g$ is called transversal if $g$ can be implemented as $U_1 \otimes \cdots \otimes U_n$ where each $U_i \in \U(2)$. Transversal gates are considered naturally fault tolerant because they do not propagate errors between physical qubits. 

Our focus will be on encoding a single logical qubit into $n$ physical qubits ($ K=2 $). In this case, the Eastin-Knill theorem shows that the transversal gates must be a finite subgroup of $\SU(2)$. The finite subgroups of $\SU(2)$ are the cyclic groups, the binary dihedral groups, and three exceptional groups. We are primarily interested in the three exceptional groups: the binary tetrahedral group $\tet$, the binary octahedral group $\oct$, and the binary icosahedral group $\ico$. These three groups correspond to the lift through the double cover $ \SU(2) \to \SO(3)$ of the symmetry groups of the tetrahedron, octahedron, and icosahedron, respectively (see \cref{fig:platonicsolids}). For more information on the finite subgroups of $ \SU(2) $, see the Supplemental Material \cite{supp}.

The group $\oct$ is better known as the single qubit Clifford group $\mathsf{C}$. Many codes implement $\oct$ transversally. For example, the $ [[7,1,3]] $ Steane code and the $[[2^{2r-1}-1,1,2^r-1]]$ quantum punctured Reed-Muller codes. More generally, all doubly even self dual CSS codes implement $ \oct $ transversally. The group $\tet$ is a subgroup of the Clifford group and there are also many codes with transversal gate group $\tet$, the most famous example being the $[[5,1,3]]$ code.

In stark contrast, no code has ever been explicitly demonstrated to implement $\ico$ transversally. This omission is particularly glaring given the role $\ico$ plays in the ``optimal absolute super golden gate set"  proposed in \cite{superGoldenGates} as the best single qubit universal gate set.

\begin{figure}[htp]
    \centering
    \includegraphics[width=8cm]{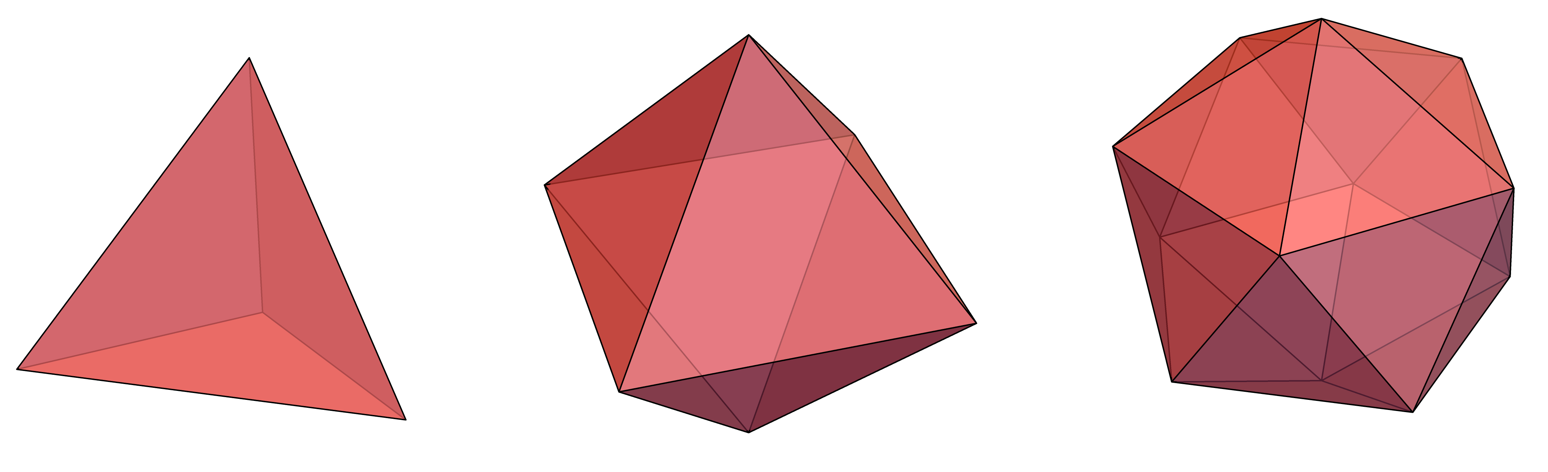}
    \caption{From left to right the following Platonic solids: tetrahedron, octahedron, icosahedron}
    \label{fig:platonicsolids}
\end{figure}

\emph{Super Golden Gates.---} A single qubit universal gate set is a finite collection of gates that generates a dense subset of $\SU(2)$. The Solovay-Kitaev theorem \cite{SKtheorem} says that a universal gate set can approximate any gate in $\SU(2)$ up to some $\epsilon$-precision using at most $\order{\log^c(1/\epsilon)}$ gates for some constant $c$ (see \cite{SK1,SK2,SK3} for bounds on $c$). Roughly: given a universal gate set, we can approximate any single-qubit gate using a relatively small number of gates. 

In the context of fault tolerance, we usually think of a universal gate set as $ \G + \tau$ where $\G$ is a finite group of gates considered ``cheap" to implement and $\tau$ is a single gate outside the group, which is considered ``expensive." This abstractly models how magic state distillation \cite{magicstatedist} works in practice; $\G$ is a set of transversal gates for some code (and so naturally fault tolerant) and $\tau$ is a gate that must be ``simulated" using magic states, distillation, and teleportation, and is usually quite costly to implement (cf.  \cite{magiccost1,magiccost2}). 

A \bfit{super golden gate set} \cite{superGoldenGates} is a universal gate set $\G + \tau$ that possesses optimal navigation properties and minimizes the number of expensive $\tau$ gates that are used (see section 2.2, 2.3 of \cite{fast2I} for a precise definition). We already know from the Solvay-Kitaev theorem that the total number of gates in any approximation will be small, but a super golden gate set in addition guarantees there won't be too many expensive $\tau$ gates. There are only finitely many super golden gate sets, including one for each of the symmetry groups of the platonic solids shown in \cref{fig:platonicsolids}.

The most familiar example of a super golden gate set is Clifford + $ T $, or equivalently, $\oct + T$. Here $T$ is the square root of the phase gate (also known as the $\pi/8$ gate). For this gate set, Clifford operations are indeed cheap since there are many codes that can implement them transversally, e.g., the $ [[7,1,3]] $ Steane code. Implementing the $ T $ gate fault tolerantly is standard in the magic state literature. As of writing, the best navigation algorithm for $\oct + T$ can efficiently factor any gate in $\SU(2)$ to within $\epsilon$ precision using at most $\tfrac{7}{3} \log_2(1/\epsilon^3)$ expensive $T$ gates (see Theorem 1 of \cite{cliffordplusTstier}).

Another example of a super golden gate set is $\ico + \tau_{60}$ defined in \cite{superGoldenGates}. Here, the cheap gates form the group $\ico$, while the expensive gate is called $\tau_{60}$. It is defined as
\[
    \tau_{60} := \tfrac{1}{\sqrt{5 \varphi + 7}} \smqty( i (2+ \varphi) & 1+i \\ -1+i & -i (2+\varphi) ), \numberthis
\]
where $\varphi = (1+\sqrt{5})/2$ denotes the golden ratio. The best navigation algorithm for $\ico + \tau_{60}$ can efficiently factor any gate in $\SU(2)$ to within $\epsilon$ precision using at most $\tfrac{7}{3} \log_{59}(1/\epsilon^3)$ expensive $\tau_{60}$ gates (see Theorem 1 of \cite{fast2I}).

Notice that the only difference between the number of $\tau$ gates in these two cases is in the base of the logarithm (which is related to the structure of the super golden gate set). Since $\log_2(x)= \log_2(59) \log_{59}(x) $, using the unviersal gate set $\ico + \tau_{60}$, instead of $\oct + T$, gives a $\log_2(59) \approx 5.9$ times reduction in the number of expensive $\tau$ gates (in the worst case). 

For example, if we want to approximate any gate in $\SU(2)$ up to a precision of $\epsilon = 10^{-10}$ then $\oct + T$ would need at most 233 $T$-gates whereas $\ico + \tau_{60}$ would only need at most $40$ $\tau_{60}$-gates. Out of all of the super golden gate sets, $\ico + \tau_{60}$ has the largest logarithm base and so it is optimal.

\emph{Summary of Results.---}
A practical implementation of the $\ico + \tau_{60}$ super golden gate set requires a cheap way to implement gates from $ \ico $. The most natural solution is to proceed as in the case of the Clifford + $ T $ super golden gate set and find quantum error-correcting codes that implement $\ico$ transversally. As already mentioned, no such code has even been demonstrated. 

In what follows, we fill this void. We first show that any code that supports $\ico$ transversally must be a non-additive code (\cref{thm:nonadditive}). Then we construct a $7$ qubit code that can correct an arbitrary error (i.e., $d =3$) and we show that it is the smallest code that can implement $\ico$ transversally (\cref{thm:smallest}). We then prove a correspondence between spin codes and multiqubit codes (\cref{lem:DickeBootstrap}) which we use to construct $d=3$ codes that implement $\ico$ transversally for all odd $n$ except $1,3,5,9,11,15,21$ (\cref{thm:family}). This result implies that the fast navigation algorithm for $\ico + \tau_{60}$ can be performed for nearly all odd numbers of qubits.  This abundance of codes is due to a symmetry phenomenon that is unique to $\ico$ among all finite subgroups of $\SU(2)$ (\cref{thm:auto1unique}).

\section{Preliminaries}
\ 
\emph{Gates.---} Single qubit quantum gates are usually presented as elements of the unitary group $\U(2)$. However, $\U(2) = e^{i \theta} \SU(2) $, so it is sufficient to consider quantum gates from the special unitary group $\SU(2)$. We will denote matrices from $\SU(2)$ by sans serif font. See \cref{tab:U2SU2gates} for our chosen correspondence. Most of the gates are standard with the notable exception of the ``facet gate" $F$ (see the Supplemental Material \cite{supp} for discussion).

\begin{table}[htp] 
 \small  
    \begin{tabular} {lll} \toprule  
     & $\U(2)$ & $\SU(2)$ \\ \toprule 
    Pauli-$X$ &  $X = \smqty( 0 & 1 \\ 1 & 0) $ & $\X = \smallminus i X$ \\ 
    Pauli-$Y$ & $Y = \smqty( 0 & i \\ \smallminus i & 0)$ & $\Y = \smallminus i Y$ \\ 
   Pauli-$Z$ & $Z = \smqty(1 & 0 \\ 0 & \smallminus 1)$ & $\Z = \smallminus i Z$ \\ \midrule  
     Hadamard & $H = \tfrac{1}{\sqrt{2}}\smqty(1 & 1 \\ 1 & \smallminus 1)$ & $\H = \smallminus i H$ \\ 
     Phase & $S = \smqty(1 & 0 \\ 0 & i)$ & $\S = e^{\smallminus i \pi /4} S$ \\ 
     Facet & $ F=H S^\dagger $  & $ \F = \tfrac{e^{\smallminus i \pi /4}}{\sqrt{2}}\smqty( 1 & \smallminus i \\ 
    1 & i) = \H \S^\dagger$ \\  \midrule  \addlinespace
   $\pi/8$-gate & $ T = \smqty(1 & 0 \\ 0 & e^{i\pi/4} )$ & $\mathsf{T} = \smqty( e^{\smallminus i \pi/8} & 0 \\ 0 & e^{i \pi /8} )$ \\ \addlinespace
   $\tfrac{2\pi}{2^r}$-Phase & $Ph(\tfrac{2\pi}{2^r}) = \smqty( 1 & 0 \\ 0 & e^{i 2\pi/2^r} )$ & $\Ph(\tfrac{2\pi}{2^r}) = \smqty( e^{\smallminus i \pi/2^r} & 0 \\ 0 & e^{i\pi/2^r} )$ \\ \addlinespace \bottomrule
    \end{tabular}
    \caption{}
    \label{tab:U2SU2gates}
\end{table}

\emph{Finite Subgroups.---} The single qubit Clifford group $\mathsf{C}$ is generated as $\mathsf{C} = \expval{ \mathsf{X}, \mathsf{Z}, \F, \H, \S }$ and has $48$ elements. This group is isomorphic to $\oct$. Although there are infinitely many (conjugate) realizations of $\oct$ in $\SU(2)$, $\expval{ \mathsf{X}, \mathsf{Z}, \F, \H, \S }$ is the only version that contains the Pauli group $\P = \expval{\X, \Z}$ and so it is the canonical choice. A subgroup of $\oct$ is the binary tetrahedral group $\tet$ with $ 24 $ elements. Again, $ \tet $ has infinitely many realizations, but the canonical choice is $\expval{ \X, \Z, \F}$, showcasing the role of the $\F$ gate. 

The binary icosahedral group $\ico$ of order $120$ also has infinitely many realizations. Unlike $\oct$ and $ \tet $ there is not a single canonical choice of $ \ico $ subgroup, but rather two, related by Clifford conjugation. One version of $ \ico $ is $ \expval{ \X, \Z, \F, \mathsf{\Phi} }$ where
\[
\mathsf{\Phi} = \tfrac{1}{2}
\smqty(
\varphi + i\varphi^{-1} & 1 \\
-1 &\varphi - i\varphi^{-1}).\numberthis
\]
The other version of $ \ico $ is $ \expval{ \X, \Z, \F, \mathsf{\Phi}^\star }$ where we obtain $ \mathsf{\Phi}^\star $ from $\mathsf{\Phi}$ by making the replacement $\sqrt{5} \to -\sqrt{5}$ and then taking the complex conjugate. 

\emph{Exotic Gates.---} The $ r+1 $ level of the (special) 1-qubit Clifford hierarchy is defined recursively as
\[
    \mathsf{C}^{(r+1)} := \{ \mathsf{U} \in \SU(2) : \mathsf{U} \P \mathsf{U}^\dagger \subset \mathsf{C}^{(r)} \}, \numberthis
\]
where $\mathsf{C}^{(1)}:=\P $ is the 1-qubit Pauli group \cite{semiClifford}. Clearly $\mathsf{C}^{(2)} $ is the one-qubit Clifford group $ \mathsf{C} $. The $\mathsf{T}$ gate is in $\mathsf{C}^{(3)}$, and in general $\Ph(2\pi/2^r)$ is in $\mathsf{C}^{(r)}$. In fact, every gate in \cref{tab:U2SU2gates} is in some level of the Clifford hierarchy. On the contrary, we have the following. 

\begin{lemma} \label{lem:exotic}
    The $\mathsf{\Phi}$ gate is not in the Clifford hierarchy. 
\end{lemma}
The essence of the proof - worked out in the the Supplemental Material \cite{supp} - is that the golden ratio $\varphi$ cannot be expressed in terms of iterated square roots of $ 2 $. Not being in the Clifford hierarchy is the sense in which we call the $\mathsf{\Phi}$ gate  \bfit{exotic}. In fact, the only gates from $\ico$ that are in the Clifford hierarchy are the gates forming the subgroup $\tet = \expval{\X, \Z, \F}$. The $ 96 $ other gates in $\ico$ are exotic (see the Supplemental Material \cite{supp}). 

On the other hand, it is known that the transversal gate group of a \textit{stabilizer} code must lie in a finite level of the Clifford hierarchy \cite{wirthmüller2011automorphisms,disjointness, transuniv}. In other words, exotic gates cannot be in the transversal gate group of a stabilizer code. This along with \cref{lem:exotic} implies our first claim.

\begin{theorem} \label{thm:nonadditive}
    Any code that implements $\ico$ transversally must be non-additive.
\end{theorem}

A quantum code is called \bfit{non-additive} if it is not equivalent via non-entangling gates to any stabilizer code. For more background about non-additive codes, see the Supplemental Material \cite{supp}.

\emph{Transversality.---} Let $ g \in \U(2) $ be a logical gate for an $((n,2,d))$ code. We say that $ g $ is \bfit{exactly transversal} if the physical gate $ g^{\otimes n} $ implements logical $ g $ on the code space. We say $ g $ is $h$-\bfit{strongly transversal} if there exists some $ h \in \U(2) $, not necessarily equal to $ g $, such that the physical gate $ h^{\otimes n} $  implements logical $ g $ on the code space. 

An $((n,2,d))$ code that implements the group $\G$ strongly transversally must transform in a $2$-dimensional faithful irrep of $\G$. For $\ico$, there are only two such irreps, the fundamental representation $\pi_2$ and the closely related representation $\overline{\pi_2}$, which is just permuted by an outer automorphism (the character table for $\ico$ can be found in the Supplemental Material \cite{supp}).

\section{The Smallest $\ico$ code}

Using a computerized search over $\ico$ invariant subspaces, we found a $((7,2,3))$ code that implements $\ico$ transversally. A normalized basis for the codespace is
\begin{align*}
    \logzero &= \tfrac{\sqrt{15}}{8} \ket*{D_0^7} + \tfrac{\sqrt{7}}{8} \ket*{D_2^7} + \tfrac{\sqrt{21}}{8} \ket*{D_4^7} - \tfrac{\sqrt{21}}{8} \ket*{D_6^7} \\
    \logone &= -\tfrac{\sqrt{21}}{8} \ket*{D_1^7} + \tfrac{\sqrt{21}}{8} \ket*{D_3^7} + \tfrac{\sqrt{7}}{8} \ket*{D_5^7} + \tfrac{\sqrt{15}}{8} \ket*{D_7^7} . \numberthis \label{code:us}
\end{align*}
Here $\ket*{D_w^n}$ is a Dicke state \cite{dicke1,dicke2,dicke3,dicke4,dicke5} defined as the (normalized) uniform superposition over all $\binom{n}{w}$ of the $n$-qubit states with Hamming weight $w$. For example,
\[
    \ket*{D_2^3} = \tfrac{1}{\sqrt{3}}\qty( \ket{011} + \ket{110} + \ket{101} ). \numberthis
\]
The weight enumerator coefficients \cite{quantumMacWilliams} of the $ ((7,2,3)) $ $ \ico $ code are
\begin{subequations}
\begin{align}
    A &= \qty(1,0,7,0,7,0,49,0), \\ 
    B &= \qty(1,0,7,42,7,84,49,66).
\end{align}
\end{subequations}
We immediately observe that the code distance is $d = 3$ since $A_i = B_i$ for each $i=0,1,2$. 

Both $X$ and $Z$ are exactly transversal since $X^{\otimes n}$ sends $\ket*{D_w^n}$ to $\ket*{D_{n-w}^n}$ and $Z^{\otimes n}$ sends $\ket*{D_w^n}$ to $(-1)^w \ket*{D_w^n}$. Thus $\X^{\otimes 7}$ implements logical $-\X$ and $\Z^{\otimes 7}$ implements logical $-\Z$. Logical $F$ is strongly-$F^*$ transversal and logical $\F$ is strongly-$\F^*$ transversal, where $ * $ denotes complex conjugation. The $[[7,1,3]]$ Steane code also implements logical $F$ and $\F$ in this way, in contrast with the $[[5,1,3]]$ code, where both $ F $ and $ \F $ are exactly transversal. Lastly, logical $\mathsf{\Phi}^\star$ is strongly-$\mathsf{\Phi}$ transversal. It follows that the code in \cref{code:us} implements $\ico$ transversally. This is the first code to have a transversal implementation of a gate outside of the Clifford hierarchy.  Note that since logical $ \mathsf{\Phi}^\star $ is being implemented, rather than logical $ \mathsf{\Phi} $, this code lives in the $\overline{\pi_2}$ irrep (as opposed to the $\pi_2$ irrep).

\begin{theorem} \label{thm:smallest}
    The smallest non-trivial code ($ d \geq 2 $) with $\ico$ strongly transversal is the $((7,2,3))$ code in \cref{code:us}. It is the unique such code in $7$ qubits. 
\end{theorem}
The proof - given in the Supplemental Material \cite{supp} - is a basic application of branching rules. The minimal error-correcting codes that implement $\tet$, $\oct$ and $ \ico $ transversally are the $ [[5,1,3]] $ code, the Steane $ [[7,1,3]] $ code, and the $ ((7,2,3)) $ code respectively. It is with this observation that the $((7,2,3))$ code should be regarded as fundamental.

\section{A Family of $\ico$ Codes}

In \cite{gross1}, the author considers the problem of encoding a qubit into a single large spin. Spin $j$ corresponds to the unique $ 2j+1 $ dimensional irrep of $ \SU(2) $, which is spanned by the eigenvectors of $J_z$ (the $z$-component of angular momentum) denoted $\ket{j,m}$ for $|m| \leq j$. Each $g \in \SU(2)$ has a natural action on the space via the Wigner $ D $ rotation operators $ D^j(g) $.

Encoding a qubit into this space means choosing a $2$-dimensional subspace. If $D^j(g)$ preserves the codespace then it will implement a logical gate. The collection of logical gates forms a finite group $\G$. We will be interested in the cases for which the logical gate is implemented as $\lambda(g)$, where $\lambda$ is an irrep of $\G$. 

We measure a spin code's performance based on how well it can correct small order isotropic errors. This is equivalent to correcting products of angular momenta. Analogous to multiqubit codes, we say a spin code has \textit{distance} $d$ if for all codewords $\logicalket{u}$ and $\logicalket{v}$ we have
\[
\logicalbra{u} J_{\alpha_1} \cdots J_{\alpha_p} \logicalket{v} = C \braket{\overline{u} }{\overline{v}} \qquad \text{for } 0 \leq p < d. \numberthis \label{eqn:KLspin}
\]
Here $J_{\alpha_i}$ is either $J_z$ or a ladder operator $J_\pm$ \cite{sakurai} and the constant $C$ is allowed to depend on $\alpha_1, \cdots, \alpha_p$ but not on the codewords. These are the Knill-Laflamme conditions (KL) for spin codes \cite{KL,gross1,gross2}.

A spin $j$ system is isomorphic to the permutationally invariant subspace of the tensor product of $n = 2j$ many spin $1/2$ systems \cite{sakurai}. An explicit isomorphism is the Dicke state mapping
\[
    \ket{j,m} \overset{\D}{\longmapsto} \ket*{D_{j-m}^{2j}}. \numberthis \label{eqn:dicke}
\]
The Dicke state mapping $\D$ behaves as an intertwiner between the natural action of $ \SU(2) $ on a spin $ j $ irrep and the natural action of $ \SU(2) $ on an $ n=2j $ qubit system via the tensor product:
\[
\D\qty[ D^j(g) \ket{j,m} ] = g^{\otimes n}  \D \ket{j,m} .  \numberthis \label{eqn:dickecovar}
\]
The main implication of this property is that $\D$ converts logical gates of a spin code into logical gates of the corresponding multiqubit code. On the other hand, the Dicke state mapping $\D$ also behaves well with respect to error-correcting properties.

\begin{lemma}\label{lem:DickeBootstrap} A spin $j$ code with distance $d = 3$ that implements logical gates from $\mathsf{G}$ corresponds under $\D$ to a permutationally invariant $\mathsf{G}$-transversal $n= 2j$ multiqubit code with distance $d=3$. 
\end{lemma}

A proof is given in the Supplemental Material \cite{supp}. This lemma means we can focus on constructing spin codes with good distance that transform in the group $\ico$. To guarantee a $d = 3$ spin code, we need to satisfy the KL conditions (\cref{eqn:KLspin}) for $p$-fold products of angular momentum where $p = 0, 1, 2$ (henceforth called the \textit{rank}). 

If the codewords are orthonormal then the rank-0 conditions are automatically satisfied. Thus we need to find codewords such that the KL conditions hold for the rank-1 errors $J_\alpha$ and for the rank-2 errors $J_\alpha J_\beta$.

Very surprisingly, the rank-1 conditions are \textit{always} satisfied when the logical group is $\ico$ and the irrep is $\overline{\pi_2}$.

\begin{lemma}[Rank-1] \label{lem:rank1} A $\ico$ spin code transforming in the $\overline{\pi_2}$ irrep satisfies all rank-$1$ KL conditions \textit{automatically}. 
\end{lemma}

On the other hand, rank-2 errors are satisfied quite generically. Similar ideas can be found in \cite{gross1,gross2,EricIanexactlytransversal}.

\begin{lemma}[Rank-2] \label{lem:rank2} Suppose a spin code implements logical $\X$ and logical $\Z$ using the physical gates $D^j(\X)$ and $D^j(\Z)$ respectively. If the codewords are real and the rank-1 KL conditions are satisfied, then the rank-2 KL conditions are also satisfied. 
\end{lemma}

The proof of each of these lemmas is given in the Supplemental Material \cite{supp}. In particular, a $\ico$ spin code with real codewords that transforms in the $\overline{\pi_2}$ irrep will satisfy both of these lemmas.

To be sure, this means that real $(\ico,\overline{\pi_2})$ spin codes have distance $d = 3$ automatically. In other words, these spin codes are deduced \textit{entirely} from symmetry. We can now use \cref{lem:DickeBootstrap} to immediately get a distance $d= 3$ multiqubit code family. 

\begin{theorem}[Family of error-correcting $\ico$ codes]\label{thm:family}
    There is an $((n,2,3))$ multiqubit code that implements $\ico$ transversally for all odd $n$ except $1,3,5,9,11,15,21$.
\end{theorem}
The exceptions are easy to understand. Only odd tensor powers of qubits branch to faithful irreps of $\ico$ and $\overline{\pi_2}$ does not appear in the permutationally invariant subspace for any of the seven odd values of $n$ listed. We give a concrete construction of the codewords in the Supplemental Material \cite{supp}.

It is worth emphasizing that rank-2 errors are satisfied fairly generically. So long as you implement the logical Pauli group and choose your codewords to be real, \cref{lem:rank2} says that you can bootstrap an error \textit{detecting} distance $d = 2$ spin code to an error \textit{correcting} distance $d = 3$ spin code for free. In contrast, the rank-1 error condition from \cref{lem:rank1} was particular to the $\overline{\pi_2}$ irrep of $\ico$. One might wonder if this restriction was unnecessary, e.g., are there any other pairs $(\G, \lambda)$ for which this automatic rank-1 condition is true? The answer is no. Only the binary icosahedral group $\ico$ affords enough symmetry.

\begin{theorem} \label{thm:auto1unique} The automatic rank-1 protection property from \cref{lem:rank1} is \textit{unique} to the pair $(\ico,\overline{\pi_2})$ among all finite subgroups of $ \SU(2) $.  
\end{theorem}

\section{Conclusion}

In this paper, we have constructed codes that implement the binary icosahedral group $\ico$ transversally, thereby completing the first half of a practical implementation of the fast icosahedral navigation algorithm.

One notable consequence of our search for a $ \ico $ code
was the discovery that $ \ico $ codes satisfying certain transversality properties  are automatically guaranteed to be error-correcting (permutationally invariant real codes transforming in the $\overline{\pi_2}$ irrep). This is the first time that the error-correcting properties of a code have been deduced purely from transversality considerations. This suggests a deep connection between transversality and error correction which requires further investigation. 

Most of the demonstrated advantage of non-additive codes over stabilizer codes has been confined to marginal improvements in the parameter $ K $ relative to fixed $ n $ and $ d $. However, our $ \ico $ code family motivates the study of nonadditive codes from a different, and much stronger, perspective. Namely, we show that nonadditive codes can achieve transversality properties which are forbidden for any stabilizer code.

\section{Acknowledgments} 
We thank J. Maxwell Silvester for helpful conversations, Michael Gullans for suggesting the use of weight enumerators in our original computerized search for the $ 7 $ qubit $ \ico $ code, Victor Albert for pointing out the relevance of the work of \cite{gross1,gross2}, Jonathan Gross for helping us rectify the misunderstandings of \cite{gross1,gross2} present in the first version of this manuscript, Sivaprasad Omanakuttan for helpful conversations in understanding the work of \cite{gross2}, Markus Heinrich for helpful conversations regarding the exoticness of the $ \Phi $ gate, and Mark Howard for helpful conversations regarding the facet gate. 

The authors acknowledge the University of Maryland supercomputing resources made available for conducting the research reported in this paper. All figures were drawn with \textsc{Mathematica 13.2}.

This research was supported in part by NSF QLCI grant OMA-2120757.

\nocite{gottesman1997stabilizer,
magicstatedist,
gidney2021stim,
faceM,
face1,
face2,
face3,
semiClifford,
diagonalCliffordHierarchy,
GF4codes,
Shadows,
schurweyl,
codeswitching,
smallestT}

\bibliography{biblio}

\ 
\newpage 
\appendix 

\makeatletter
\renewcommand{\theequation}{S\arabic{equation}}
\renewcommand{\thetable}{S\arabic{table}}
\renewcommand{\thefigure}{S\arabic{figure}}
\renewcommand{\thelemma}{S\arabic{lemma}}
\setcounter{table}{0}
\setcounter{figure}{0}
\setcounter{lemma}{0}
\setcounter{equation}{0}

\ \newpage 
\section{\large Supplemental Material}

\section{Discussion of the Facet Gate}\label{sec:facet} 
The defining feature of the gate we have dubbed $F$ is that its conjugation action on the Paulis is by cycling:
\[
\X \to \Y, \quad \Y \to \Z, \quad \Z \to \X. \numberthis
\] 
This gate used to be referred to as the ``$T$" gate (cf. \cite{gottesman1997stabilizer,magicstatedist}) but nowadays $T$ is generally reserved for the $\pi/8$  gate (see \cref{tab:U2SU2gates}). 

In Stim \cite{gidney2021stim}, $F$ is referred to as $C_{xyz}$ and called an ``axis cycling gate" but unfortunately, the letter ``$C$" has quite a bit of overlap with other terms in the quantum information literature. We referred to $F$ as ``$M$" in earlier versions in this manuscript because \cite{faceM} used ``$M_3$" for this gate. 

On the other hand, it seems that in some recent contexts, \cite{face1,face2,face3}, $F$ has been called either the ``facet" or ``face" gate since $F$ is among $8$ gates that are Lie-generated by the vectors orthogonal to the $8$ facets/faces of the stabilizer polytope. We prefer ``facet" since ``face" and ``phase" (the $S$ gate) are near-homophones.

\section{Finite Subgroups of $ \SU(2) $}

\newcommand{\Ad}{\mathrm{Ad}}

In the adjoint representation, $ \SU(2) $ acts by conjugation on its three dimensional Lie algebra $\mathfrak{su}(2)$ 
giving the \textit{adjoint map}
\[
\Ad: \SU(2) \mapsto \SO(3). \numberthis 
\]
This map is ``2-to-1" in the sense that both $ \pm g \in \SU(2) $ have the same conjugation action and so map to the same element of $\SO(3)$. The map $\Ad$ assigns each $g \in \SU(2)$ to a $ 3 \times 3 $ matrix whose columns are determined by the action on the three Paulis $ X,Y,Z $. For example, conjugation by the facet gate $ \F $ sends $ X \to Y$, $Y \to Z$, and $Z \to X $, so $ \Ad(\F)= \smqty(0 & 0 & 1 \\ 1 & 0 & 0 \\0 & 1 & 0 ) $. Taking the inverse image under the adjoint map gives an order two lift that associates every finite subgroup $ \G \subset \SO(3) $ to a subgroup $ \Ad^{-1}(\G) \subset \SU(2) $ of twice the size.

\begin{table}[htp]
    \centering 
    \begin{tabular} {lll|lll}
    \toprule 
     \multicolumn{3}{c}{$\SO(3)$} & \multicolumn{3}{c}{$\SU(2)$} \\ \midrule 
     $ \mathrm{C}_n $ & $ \Ad( \Ph\qty(\tfrac{2 \pi}{n}))$ & $n$ & $ \mathrm{C}_{2n} $ &  $ \Ph(\tfrac{2 \pi}{n}) $ & $2n$ \\
     $ \mathrm{Dih}_n $ & $ \Ad(\X), \Ad(\Ph(\tfrac{2 \pi}{n}))$ & $2n$ & $ \mathrm{BD}_n $ &  $ \X, \Ph(\tfrac{2 \pi}{n}) $ & $4n$ \\
     $ \mathrm{A}_4 $ & $ \Ad(\X), \Ad(\Z), \Ad(\F)$ & $12$ & $ \tet  \cong \SL(2,3) $ &  $ \X, \Z, \F $ & $24$ \\ 
       $ \mathrm{S}_4 $ & $ \Ad(\H), \Ad(\S) $ & $24$ & $ \oct $ &  $ \X, \Z, \F $ & $48$ \\ 
       $ \mathrm{A}_5 $ & $ \Ad(\F), \Ad(\Phi) $ & $60$ & $ \ico \cong  \SL(2,5) $ &  $ \F, \Phi $ & $120$ \\  \bottomrule 
    \end{tabular}
    \caption{Name, generating set, and order for each finite subgroup of $ \SO(3) $ and for the corresponding even order subgroup of $ \SU(2) $ }
    \label{tab:finitesubgroupsSO3andSU2}
\end{table}

Some common notation for finite groups is used here: $ \mathrm{C}_n $ are cyclic groups, $ \mathrm{Dih}_n $ are dihedral groups, $ \mathrm{BD}_n $ are binary dihedral groups, $ \mathrm{A}_4 $ and $ \mathrm{A}_5 $ are alternating groups, $ \mathrm{S}_4 $ is a symmetric group, and $ \SL(2,p) $ is the group of $ 2 \times 2 $ determinant $ 1 $ matrices with entries from $ \mathbb{F}_p $, the field with $ p $ elements.

Each row of the table corresponds to the rotational symmetries of some rigid body. The last three rows of \cref{tab:finitesubgroupsSO3andSU2} correspond to the tetrahedron, octahedron, and icosahedron respectively.

\section{Proof of \cref{lem:exotic}}\label{sec:proofexotic}

\begin{proof}
All gates in the single qubit Clifford hierarchy are semi-Clifford \cite{semiClifford}. Let $ \mathsf{U} $ be a (determinant $ 1 $) semi-Clifford gate. Then we can write $\mathsf{U}= \mathsf{V}_1 \mathsf{D} \mathsf{V}_2$ where $ \mathsf{D} $ is a (determinant $ 1 $) diagonal single qubit gate and $ \mathsf{V}_1, \mathsf{V}_2 $ are in the (special) Clifford group $ \mathsf{C}$.
Suppose that $ \mathsf{U} \in \mathsf{C}^{(r)} $, then we must have $ \mathsf{D} \in \mathsf{C}^{(r)} $, by Proposition 3 of \cite{semiClifford}.  

Then \cite{diagonalCliffordHierarchy} shows that $ \mathsf{D} \in \mathsf{C}^{(r)} $ has all diagonal entries some power of the root of unity $ \zeta_{2^{r+1}}=e^{2 \pi i/2^{r+1}} $. All entries of $ \mathsf{V}_1, \mathsf{V}_2 $ are in the cyclotomic field $ \mathbb{Q}(\zeta_8)=\mathbb{Q}(\sqrt{2},i)$. Thus any (determinant 1) single qubit gate in $ \mathsf{C}^{(r)} $ must have all its entries in $ \Q(\zeta_{2^{r+1}}) $. However $ \sqrt{5} $ is not in any $ \Q(\zeta_{2^{r+1}}) $. The result follows.
\end{proof}

\section{Most gates in $\ico$ are exotic} \label{sec:exoticproof96}

\begin{proof}
The left cosets of $\tet$ in $\ico$ are $g \cdot \tet$ for $g \in \ico$. These form a partition of $\ico$ and can be labeled by five representatives since $|\ico|/|\tet| =  120/24 = 5$. One choice of the five coset representatives is $\indicator$, $\mathsf{\Phi}$, $\mathsf{\Phi}^2$, $\mathsf{\Phi}^3$, and $\mathsf{\Phi}^4$, since $\mathsf{\Phi}^5 = -\indicator$. The 24 elements in the identity coset $\indicator \cdot \tet$ are in the Clifford group while the $96$ elements in the other four cosets are clearly outside the Clifford hierarchy because each contains an entry that includes a $\sqrt{5}$.
\end{proof}

\section{$\ico$ Character Table}\label{sec:2Itable}

The character table for $\ico$ is given in \cref{tab:chartable}. The irrep $ \pi_i $ denotes the restriction to $ \ico $ of the unique $ i $ dimensional irrep of $ \SU(2) $. An overbar denotes the image of an irrep under the outer automorphism of $ \ico $. The outer automorphism acts on characters by taking $ \sqrt{5} \to -\sqrt{5} $, and equivalently $\varphi \to -\varphi^{-1} $. We use primes to distinguish other irreps of the same dimension. In all cases, the subscript of the irrep is the dimension.

\begin{table}[htp]
    \centering
    \begin{tabular}{cccccccccc} \toprule 
   Class & $[\indicator]$ & $[-\indicator]$ & $[\X,\Y,\Z]$ & $[\F]$ & $[-\F]$ & $[\mathsf{\Phi}]$ & $[\mathsf{\Phi}^2]$ & $[\mathsf{\Phi}^3]$ & $[\mathsf{\Phi}^4]$ \\
    Size  & 1 & 1 & 30 & 20 & 20 & 12 & 12 & 12 & 12 \\ 
   Order  & 1 & 2 & 4 & 3 & 6 & 10 & 5 & 10 & 5 \\ \midrule 
    $\pi_1$ & 1 & 1 & 1 & 1 & 1 & 1 & 1 & 1 & 1 \\ 
    $\pi_2$ & 2 & -2 & 0 & 1 & -1 &  $\varphi$ & $\varphi^{-1}$ & $-\varphi^{-1}$ & $-\varphi$ \\
    $\overline{\pi_2 }$ & 2 & -2 & 0 & 1 & -1 & $-\varphi^{-1}$ & $-\varphi$ & $\varphi$ & $\varphi^{-1}$   \\ 
    $\pi_3$ & 3 & 3 & -1 & 0 & 0 &  $\varphi$ &  $-\varphi^{-1}$ & $-\varphi^{-1}$ &  $\varphi$  \\ 
     $\overline{\pi_3}$ & 3 & 3 & -1 &  0 & 0 &  $-\varphi^{-1}$ & $\varphi$ & $\varphi$ & $-\varphi^{-1}$ \\ 
     $\pi_4$ &  4 & -4 & 0 & -1 & 1 & 1 & -1 &  1 & -1  \\ 
     $\pi_{4'}$ & 4 & 4 & 0 & 1 & 1 & -1 & -1  & -1 & -1 \\ 
     $\pi_5$ & 5 & 5 & 1 & -1 & -1 & 0 & 0 &  0 & 0\\ 
    $\pi_6$ & 6 & -6 & 0 & 0 & 0 & -1 & 1 & -1 & 1 \\  \bottomrule 
    \end{tabular}
    \caption{$2 I$ Character Table,  $[g]$ is the conjugacy class containing the element $g$}
    \label{tab:chartable}
\end{table}

There are only two irreps with dimension $2$: the fundamental representation $\pi_2$, and the closely related irrep $\overline{\pi_2}$, which is just $ \pi_2 $ permuted by an outer automorphism. Any single qubit ($K=2$) code that implements $\ico$ strongly transversally must transform in one of these two irreps.

\section{Non-Additive Codes}

Most of today's quantum error correction (QEC) uses \textit{stabilizer} codes. The codespace of a stabilizer code is the simultaneous $ +1 $ eigenspace of a subgroup of the $n$-qubit Pauli group. In the early days of QEC, it was shown \cite{GF4codes} that stabilizer codes were related to ``additive" classical codes over $\mathrm{GF}(4)$, the finite field of 4 elements. Here additive means closed under addition (as opposed to linear which means closed under both addition and scalar multiples). For this reason, codes that could not be realized as stabilizer codes were dubbed ``non-additive" and the nomenclature seems to have stuck (even though ``non-stabilizer codes" might be preferred).  

The first error detecting ($d=2$) non-additive code was discovered by Rains, Hardin, Shor, and Sloane in 1997 \cite{nonadditive}. Later in the same year, Roychowdhury and Vatan discovered the first error correcting ($d=3$) non-additive code \cite{roychowdhury1997structure}. 


The codewords of a stabilizer code are stabilizer states, and moreover can can always be brought into a form for which the coefficients with respect to the computational basis are all $\pm 1$ (up to a uniform normalization). However, not all codes with codewords of this form are stabilizer codes. These types of codes are called \textit{codeword stabilized codes} (CWS) \cite{CWS}. In some ways, these are the ``least" non-additive codes. CWS codes are equivalent to graphical quantum codes \cite{graphical} and this approach was used in \cite{nonadditive2} to find a non-additive error correcting code that outperformed the best known stabilizer code of the same length. 

Non-additive codes that are not CWS are not well understood, except in a few specific cases. One example is the $XS$ and $XP$ formalism in which the stabilizer group is generalized to include certain diagonal non Pauli gates \cite{XScodes,XPcodes}. It would be interesting to see if the codes we developed in the main text could be formulated as either $XS$ or $XP$ stabilizer codes (we suspect not since $\ico$ doesn't even have an $\S$ transversal gate). 

In a completely different direction, there are the GNU and shifted-GNU codes \cite{gnuCodes,gnuCodesShifted}. Similar to the codes in this paper, GNU type codes are permutationally invariant and thus are necessarily non-additive by the work of \cite{automorph}. Recently, a formalism has been described that allows the efficient decoding of any permutationally invariant codes \cite{ouyangPIcodesQECC}. Both the GNU codes and the codes in this paper fall under this formalism.

\section{Proof of \cref{thm:smallest}}\label{sec:proofsmallest}

\begin{proof} 
The irreps of $\SU(2)$ are labelled by $j$, where $j$ is either integral or half-integral. The half-integral irreps are faithful and the integral irreps are not faithful. 

A qubit is spin $\tfrac{1}{2}$ and so $n$ qubits live in the Hilbert space $\tfrac{1}{2}^{\otimes n}$. This reducible representation of $\SU(2)$ splits into a direct sum of irreps of $\SU(2)$. When $n$ is even it splits into only integral spin irreps (non-faithful) and when $n$ is odd it splits into only half-integral spin irreps (faithful).

If a code is strongly $\ico$-transversal then it must live in either a $\pi_2$ or $\overline{\pi_2}$ irrep within $\frac{1}{2}^{\otimes n}$. Because both $\pi_2$ and $\overline{\pi_2}$ are faithful, they can only exist in tensor powers for odd $n$. 

It is well known that there are no non-trivial codes for $n \leq 3$ qubits \cite{GF4codes,Shadows}. Given the above restriction on the parity of $n$, the first time a non-trivial strongly $\ico$-transversal code could appear is in $5$ qubits. The Hilbert space of $5$ qubits branches into irreps of $\ico$ as
\[
    \tfrac{1}{2}^{\otimes 5} = \underbrace{\pi_6}_{5/2} + \underbrace{4 \pi_4}_{3/2} + \underbrace{5 \pi_2}_{1/2},\numberthis
\]
where we have decomposed with respect to spin using the Schur-Weyl duality \cite{schurweyl}. We first notice there is no $\overline{\pi_2}$ copy. Moreover, we see that every $\pi_2$ lives entirely within the spin $1/2$ irrep. This means that these are actually representations of the entire group $\SU(2)$ and so are $\SU(2)$ strongly transversal. But that means there are infinitely many transversal gates so by the Eastin-Knill theorem these must be trivial codes ($ d=1 $).

In $7$ qubits the branching is
\[
    \tfrac{1}{2}^{\otimes 7} = \underbrace{\overline{\pi_2} + \pi_6}_{7/2} + \underbrace{6 \pi_6}_{5/2} + \underbrace{14\pi_4}_{3/2} + \underbrace{14 \pi_2}_{1/2}. \numberthis
\]
Again one can see that all $\pi_2$ irreps are within the spin $\frac{1}{2}$ irrep of $\SU(2)$ and so are $\SU(2)$ strongly transversal and must be trivial codes. On the other hand, there is a unique $\overline{\pi_2}$ within the totally symmetric subspace (spin $7/2$). This is the $((7,2,3))$ code from \cref{code:us}. This proves existence, uniqueness, and minimality. 
\end{proof}

\subsection{Relation to Previous Work}

\emph{Relation to \cite{2004permutation}.---}
The $((7,2,3))$ $\ico$ code is a linear combination of Dicke states and so it is permutationally invariant. As a result, the relevance of the work in \cite{2004permutation} came to our attention. The authors display two permutationally invariant $((7,2,3))$ codes, both of which are equivalent via non-entangling gates to the code in \cref{code:us}. In fact, the codes therein were the first examples noticed in \cite{CWS} of non-additive codes that are not CWS.

\emph{Relation to \cite{gross1}.---} Applying \cref{lem:DickeBootstrap} to the $ j=7/2$,  $\G=\ico $ spin code constructed in \cite{gross1} gives a multiqubit code which is equivalent via non-entangling gates to the code in \cref{code:us}. 

\emph{Relation to \cite{gross2}.---} Ideas related to how we constructed the $\ico$ family were used in \cite{gross2} to construct $\oct$ transversal multiqubit codes. However, the results of \cref{thm:auto1unique} show that \cref{lem:rank1} (satisfying rank-1 errors automatically) is completely unique to $\ico$.

\section{Proof of \cref{lem:DickeBootstrap}} \label{app:dicke}

\emph{Logical Gates.---} Let's start by proving the intertwining relation found in \cref{eqn:dickecovar}. If we can show this equation is true at the Lie algebra level then it will automatically be true at the Lie group level via exponentiation. A convenient basis for the complexified algebra $\mathfrak{su}_\C(2)$ is $(J_+, J_-, J_z)$ where $J_z$ is the $z$-component of angular momentum and $J_\pm = J_x \pm i J_y$ are ladder operators. It thus suffices to show
\[
    \D\qty[ J_\alpha \ket{j,m} ] = \sum_{i=1}^n J_\alpha^{(i)} \D \ket{j,m}, \numberthis \label{eqn:dickebasis}
\]
where $ J_\alpha$ is either $J_+$, $J_-$, or $J_z$ and $J_\alpha^{(i)}$ is $J_\alpha$ on the $i$-th qubit.

Let's start with the $J_z$ condition. On the left side of \cref{eqn:dickebasis} we have
\begin{align}
    \D \qty[ J_z \ket{j,m} ] &= m \D \ket{j,m} =m \ket*{D^{2j}_{j-m}}. \label{eqn:dickeproof1}
\end{align}
where we have used the fact that $J_z \ket{j,m} = m \ket{j,m}$. On the right hand side of \cref{eqn:dickebasis} we have
\begin{align}
    \sum_{i=1}^n J_z^{(i)} \D \ket{j,m} &= \tfrac{1}{2}\sum_{i=1}^n Z^{(i)} \ket*{D^{2j}_{j-m}} \\
    &= \tfrac{-(j-m) + (j+m)}{2} \ket*{D^{2j}_{j-m}} \label{eqn:dickeproof2} \\ 
    &= m \ket*{D^{2j}_{j-m}}.
\end{align}
In \cref{eqn:dickeproof2} we have used the fact that there are $(j-m)$ many $\ket{1}$ kets and $2j-(j-m)=j+m$ many $\ket{0}$ kets. Then we used the fact that $Z\ket{0} = \ket{0}$ and $Z\ket{1} = - \ket{1}$. This calculation and \Cref{eqn:dickeproof1} prove \cref{eqn:dickebasis} for $J_z$. 

Now let's move on to the ladder operators $J_\pm$. To start, recall \cite{sakurai}
\[
J_\pm \ket{j,m} = \sqrt{(j \mp m)(j\pm m+1)} \ket{j,m\pm 1}. \numberthis \label{eqn:ladders}
\]
Then the left hand side of \cref{eqn:dickebasis} is
\begin{align}
    \D\qty[ J_\pm \ket{j,m}] &= \sqrt{(j \mp m)(j\pm m+1)} \D \ket{j,m\pm 1} \\
    &= \sqrt{(j \mp m)(j\pm m+1)} \ket*{D^{2j}_{j-m \mp 1}} \label{eqn:dickeproof3}
\end{align}

Now let's consider the right hand side of \cref{eqn:dickebasis}. In the spin $1/2$ case
\[
    J_+ = \smqty(0 & 1 \\ 0 & 0), \quad J_- = \smqty(0 & 0 \\ 1 & 0). \numberthis
\]
However, because $\ket{0} = \ket*{\tfrac{1}{2},\tfrac{1}{2}}$ and $\ket{1} = \ket*{\tfrac{1}{2},-\tfrac{1}{2}}$ then counter-intuitively we have $J_+ \ket{0} = 0$, $J_+ \ket{1} = \ket{0}$, $J_- \ket{0} = \ket{1}$, and $J_- \ket{1} = 0$. To compute the right hand side of \cref{eqn:dickebasis} we need the following lemma.
\begin{lemma} \label{lem:rightside}
\begin{align}
    \sum_i J^{(i)}_-  \ket*{D^n_w} &= \sqrt{(n-w)(w+1)} \ket*{D^n_{w+1}}, \\
    \sum_i J^{(i)}_+  \ket*{D^n_w} &= \sqrt{w(n-w+1)} \ket*{D^n_{w-1}}.
\end{align}
\end{lemma}
\begin{proof}
    A Dicke state can be expanded as
    \[
        \ket*{D^n_{w}} = \tfrac{1}{\sqrt{\binom{n}{w}} }\sum_{ wt(s) = w } \ket{s}.
    \]
     The sum is over all length $n$ bit strings of Hamming weight $w$.
    Then
    \begin{align}
        \sum_{i=1}^n J^{(i)}_-  \ket*{D^n_w} &=  \tfrac{1}{\sqrt{\binom{n}{w}} } \sum_{i=1}^n \sum_{ wt(s)=w}  J^{(i)}_- \ket{s} \\
        &\explainequals{1} \tfrac{1}{\sqrt{\binom{n}{w}} } \sum_{i=1}^n \sum_{\substack{wt(s')=w+1 \\ s'_i = 1}  }  \ket{s'}  \\
        &\explainequals{2}  \tfrac{(w+1)}{\sqrt{\binom{n}{w}} } \sum_{wt(s')=w+1 }  \ket*{s'} \\
        &= \sqrt{(n-w)(w+1)}\ket*{D^n_{w+1}}
    \end{align}
    Line (1) follows from the fact that $ J_-^{(i)} $ annihilates any $ \ket{s} $ with $ s_i=1 $, while if $ s_i=0 $, then $ J_-^{(i)} $ just changes $ s_i $ to a $ 1 $.
     For line (2), note that the sum $ \sum_{i=1}^n \sum_{wt(s')=w+1, s'_i = 1  } $ is over exactly $ n \binom{n-1}{w} $ terms, and every weight $ w+1 $ bit string  $ s' $ appears in the sum the same number of times. Since there are exactly $ \binom{n}{w+1} $ many weight $ w+1 $ bit strings, and $ n \binom{n-1}{w}=(w+1) \binom{n}{w+1} $, that accounts for the factor of $ w+1 $ in line (2).

    The proof of $J_+$ is similar except the compensatory quantity is $(n-w+1)$ instead of $w+1$, stemming from the identity $ n \binom{n-1}{w-1}=(n-w+1) \binom{n}{w-1} $.
   
\end{proof}
If we plug in $n = 2j$ and $w = j-m$ to \cref{lem:rightside} then we can compute the right hand side of \cref{eqn:dickebasis} and see that it agrees with \cref{eqn:dickeproof3} as desired. This concludes the proof that logical gates are preserved under $\D$.

\emph{Distance.---}  Let $E$ be a Pauli string. There is a natural action of the permutation $\sigma \in \mathrm{S}_n$ (the symmetric group on $n$ letters) on $E$ by permutation of the tensor factors. For example, if $E = XZI$ then $(123) \cdot E = IXZ$. For a general operator $A$, we define the action of $\sigma \in \mathrm{S}_n$ on $A$ by $\sigma \cdot A = P_\sigma^{\dagger} A P_\sigma$ where $P_\sigma$ is the $ 2^n \times 2^n $ permutation matrix corresponding to permuting the $ n $ tensor factors by the permutation $\sigma \in \mathrm{S}_n $.

Let's start with the following lemma regarding permutationally invariant multiqubit codes. We call $ E $ a permutationally invariant error of weight $ w $ if $ \sigma \cdot E = E $ for all $ \sigma \in \mathrm{S}_n $ and $ E $ is a linear combination of weight $ w $ Pauli errors.


\begin{lemma}\label{lem:perm} If a permutationally invariant multiqubit code satisfies the KL conditions for a basis of permutationally invariant errors of weight $ w $, then the KL conditions are satisfied for \textit{all} errors of weight $ w $.
\end{lemma}
\begin{proof}
 Let $\Pi$ be a multiqubit code projector. Suppose the KL conditions are satisfied 
    \[
        \Pi E \Pi = c_E \Pi, \numberthis \label{eqn:KLpermpart}
    \]
    for all permutationally invariant errors of weight $w$. 
    If the code is permutationally invariant then $P_\sigma \Pi = \Pi P_\sigma = \Pi$ for all $ \sigma \in \mathrm{S}_n $. Let $E$ be an arbitrary Pauli error of weight $w$. Then
    \[
        \Pi E \Pi = \Pi P_\sigma^\dagger E P_\sigma \Pi = \Pi E_\sigma \Pi \qquad \forall \sigma \in \mathrm{S}_n, \numberthis \label{eqn:symm}
    \]
    where $E_\sigma = P_\sigma^\dagger E P_\sigma$. Define $\mathrm{Sym}(E) = \tfrac{1}{n!} \sum_{\sigma \in \mathrm{S}_n} E_\sigma$. It follows that
    \[
       \Pi E \Pi =  \Pi \mathrm{Sym}(E) \Pi.  \numberthis
    \]
    Because $\mathrm{Sym}(E)$ is both permutationally invariant and has weight $w$, we can use \cref{eqn:KLpermpart} to conclude that 
    \[
        \Pi E \Pi = c_{\mathrm{Sym}(E)} \Pi , \numberthis
    \]
    for \textit{any} Pauli string $E$ of weight $w$.
    So we can take the KL conditions for permutationally invariant errors and immediately bootstrap them to satisfy KL conditions for all the non permutationally invariant errors $ E $.
    \end{proof}

Let $\ket{u} = \sum_m c_m \ket{j,m}$ be a spin codeword and let $\ket{\widetilde{u}} = \D \ket{u}$ be the corresponding permutationally invariant multiqubit codeword. Then \cref{eqn:dickebasis} says
\[
    \D\qty[ J_\alpha \ket{u} ] =  \widetilde{J_\alpha} \ket{\widetilde{u}}, \numberthis
\]
where we have defined $\widetilde{J_\alpha} := \sum_i J_\alpha^{(i)}$. Notice that $\widetilde{J_\alpha}$ is weight $1$ and is permutationally invariant. Therefore we see that the Dicke stae mapping $\D$ has induced an isomorphism between spin errors and weight 1 permutationally invariant multiqubit errors:
\[
    J_\alpha \longmapsto  \widetilde{J_\alpha}. \numberthis
\]
Because $\D$ is unitary, it preserves the inner product. In other words, if a spin code has distance $2$, i.e., if $\bra{v} J_\alpha \ket{u} = c_\alpha \braket{u}{v} $, then the corresponding multiqubit code satisfies
\[
    \bra{\widetilde{v}}  \widetilde{J_\alpha} \ket{\widetilde{u}} =  c_\alpha \braket{\tilde{u}}{\tilde{v}} . \numberthis \label{eqn:weightonedicke}
\]
But $\widetilde{J_\alpha}$ form a basis for the 3-dimensional weight 1 permutationally invariant errors. So by \cref{lem:perm}, if a spin code has distance $d = 2$ then the corresponding multiqubit code under $\D$ also has distance $d = 2$.

\ \\
Similarly, $\D$ induces an isomorphism 
\[
    J_\alpha J_\beta  \longmapsto \widetilde{J_\alpha } \widetilde{J_\beta} =\Big( \sum_{i_1} J_\alpha^{(i_1)} \Big) \Big( \sum_{i_2} J_\beta^{(i_2)}\Big). \numberthis
\]
Again, because $\D$ is unitary, it follows that if a spin code has distance $3$, i.e., $\bra{v} J_\alpha J_\beta \ket{u} = c_{\alpha \beta} \braket{u}{v}$, then 
\[
    \bra{\widetilde{v}}  \widetilde{J_\alpha} \widetilde{J_\beta} \ket{\widetilde{u}} =  c_{\alpha \beta} \braket{\tilde{u}}{\tilde{v}}. \numberthis \label{eqn:weighttwodicke}
\]
Clearly $\widetilde{J_\alpha} \widetilde{J_\beta}$ is permutationally invariant. However, $\widetilde{J_\alpha } \widetilde{J_\beta}$ will not always be a weight 2 error, it will most likely be a mix of errors up to weight $2$. For example,  
\[
    \widetilde{J_+} \widetilde{J_-} = \underbrace{\sum_{i_1 \neq i_2} J_+^{(i_1)} J_-^{(i_2)}}_\text{weight 2} + \underbrace{\sum_i J_+^{(i)} J_-^{(i)} }_\text{weight 1}. \numberthis
\]
However, if a spin code is distance $3$, then \cref{eqn:weightonedicke} still holds. This together with \cref{eqn:weighttwodicke} implies that 
\[
\bra{\widetilde{v}}  \sum_{i_1 \neq i_2} J_+^{(i_1)} J_-^{(i_2)} \ket{\widetilde{u}} =  C \braket{\tilde{u}}{\tilde{v}}. \numberthis
\]
This is an honest weight 2 permutationally invariant KL error condition. We can perform the same trick with the other errors $\widetilde{J_\alpha} \widetilde{J_\beta}$.

Thus \cref{eqn:weighttwodicke} holds for a basis of permutationally invariant weight 2 errors and so \cref{lem:perm} implies that the corresponding multiqubit code has distance $d = 3$.

This completes the proof of \cref{lem:DickeBootstrap}.

\section{Proof of \cref{lem:rank1}}\label{sec:proofAPP}

\begin{proof}
A spin $j$ code that transforms in a faithful 2-dimensional irrep $\lambda$ of $\mathsf{G}$ is defined via the projector
\[
    \Pi_{(\G,\lambda)} := \frac{2}{|\G|} \sum_{g \in \G } \chi_\lambda^*(g) D^j(g). \numberthis
\]
A rank-1 error is written $J_\alpha$ but we can also write it more conveniently as $\rho^j(E)$ where $E$ is an element of the abstract Lie algebra $\mathfrak{su}_{\mathbb{C}}(2)$. Then we will show
\[
    \Pi_{(\ico,\overline{\pi_2})} \rho^j(E) \Pi_{(\ico,\overline{\pi_2})} = 0 \numberthis
\]
which implies that the rank-1 KL conditions are satisfied. 

Dropping the global scalar factor of $ \tfrac{2}{|G|} $ for convenience, we have
\begin{align}
    &\Pi_\lambda \rho^j(E) \Pi_\lambda \\
    &= \qty(\sum_{g \in \G} \chi_\lambda^*(g) D^j(g) ) \rho^j(E) \qty(\sum_{h \in \G} \chi_\lambda^*(h)  D^j(h) ) \\ 
    &= \sum_{g,h \in \G} \chi_\lambda^*(g) \chi_\lambda^*(h) \qty{ D^j(g) \rho^j(E) D^j(h) } \\
    &= \sum_{g' \in \G}  \sum_{g \in \G} \chi_\lambda^*(g) \chi_\lambda^*(g^{-1} g')  \rho^j(g E g^{-1}) D^j(g') \\
    &= \sum_{g' \in \G} \Gamma^j(g',E) D^j(g').
\end{align}
In the penultimate line, we have used the fact that $\rho^j$ is the Lie derivative of $D^j$ and so $D^j(g) \rho^j(E) D^{j \dagger}(g)=\rho^j(g E g^{-1})$. And $ \Gamma^j(g',E) $ is the quantity
\begin{align}
    \Gamma^j(g',E) &= \sum_{g \in \G} \chi_\lambda^*(g) \chi_\lambda^*(g^{-1} g') \rho^j(g E g^{-1}) \\
    &=  \rho^j\qty( \sum_{g \in \G} \chi_\lambda^*(g) \chi_\lambda^*(g^{-1} g') g E g^{-1} ).
\end{align}
Here we are using the linearity of $\rho^j$ (since it's a Lie algebra representation). Now what this says is that we can compute the representation-independent quantity 
\[
\widetilde{\Gamma}(g',E) := \sum_{g \in \G} \chi_\lambda^*(g) \chi_\lambda^*(g^{-1} g') g E g^{-1}, \numberthis \label{eqn:calculation}
\]
which is an element of $\mathfrak{su}_\mathbb{C}(2)$, then lift the result to a $2j+1$ dimensional representation using $\rho^j(\widetilde{\Gamma}) = \Gamma^j$. In practical terms, one can compute \cref{eqn:calculation} in the fundamental spin $1/2$ representation and the same result will be true in \textit{all} representations. 

Let's now specialize to $ \G= \ico $ and $\lambda = \overline{\pi_2}$. We can check with a computer that
\[
    \widetilde{\Gamma}_{(\ico,\overline{\pi_2})}(g',E) = \sum_{g \in \ico} \chi_{\overline{\pi_2}}^*(g) \chi_{\overline{\pi_2}}^*(g^{-1} g') g E g^{-1} = 0 \numberthis 
\]
for every $g' \in \ico$ and for every $E \in \mathfrak{su}_{\mathbb{C}}(2)$. Thus it follows that $\Gamma^j = \rho^j(0) = 0$ for every $j$ and so $\Pi_\lambda \rho^j(E) \Pi_\lambda = 0$ for every $j$ and for all $E \in \mathfrak{su}_\mathbb{C}(2)$. This proves the desired result. 

Although this is a computer assisted proof, it should still be considered mathematically rigorous. We have reduced the problem from being representation dependent (depending on $j$) to one that is representation-independent (not depending on $j$) and thus we only needed to check a finite number of cases: the $ 120 $ elements of $ \ico $ and some choice of a $ 3 $ element basis for $\mathfrak{su}_\mathbb{C}(2)$ (for example, $J_\pm$ and $J_z$). 
\end{proof}

\section{Proof of \cref{lem:rank2}}

Write $J_1 = J_+$, $J_0 = J_z$, and $J_{-1} = J_-$. Instead of writing the physical gates as $D^j(\X)$ and $D^j(\Z)$, we will just write $\X$ and $\Z$. To be sure, these are $2j+1$ dimensional representations of a $\pi$ rotation around the $x$ and $z$ axes respectively. Then $X = i \X$ and $Z = i \Z$. Note that \cite{sakurai} 
\begin{align}
    X^\dagger J_\alpha X &= (-1)^{\alpha +1} J_{-\alpha } \\
    Z^\dagger J_\alpha Z &= (-1)^\alpha J_\alpha.
\end{align}
Also $J_\alpha$ is real and $J_\alpha^\dagger = J_{- \alpha }$.

\begin{lemma}\label{lem:covarsymmetric} Suppose a spin $j$ code has codewords $\ket{0}$ and $\ket{1}$. Also suppose $X \ket{0} = \ket{1}$, $X \ket{1} = \ket{0}$, $Z \ket{0} = \ket{0}$, and $Z \ket{1} = - \ket{1}$. If the spin codewords are real then all symmetric rank-2 errors automatically satisfy the KL conditions. 
\end{lemma}
\begin{proof}
We can write the assumptions on the codewords more concisely as $X \ket{u} = \ket{u+1}$ and $Z \ket{u} = (-1)^u \ket{u}$ where $u$ and $v$ are labels in $\{0,1\}$ and addition is taken modulo 2. Then
    \begin{align}
        \bra{u} J_\alpha J_\beta \ket{v} &= (-1)^{u+v} \bra{u} \qty( Z^\dagger J_\alpha Z ) \qty( Z^\dagger J_\beta Z ) \ket{v} \\
        &= (-1)^{u+v} (-1)^{\alpha + \beta } \bra{u} J_\alpha J_\beta \ket{v} \\
        &\explainequals{1} (-1)^{u+v} (-1)^{\alpha + \beta }  \bra{v} J_\beta^\dagger J_\alpha^\dagger \ket{u} \\
         &=(-1)^{u+v} (-1)^{\alpha + \beta }  \bra{v} J_{-\beta} J_{-\alpha } \ket{u} \\
         &=(-1)^{u+v} \bra{v} \qty( X^\dagger J_{\beta} X) \qty( X^\dagger J_{\alpha } X) \ket{u} \\
         &= (-1)^{u+v} \bra{v+1} J_\beta J_\alpha \ket{u+1}.
    \end{align}
\end{proof}

where line (1) uses the realness of the code words (and the realness of the $ J_\alpha $).

To use this lemma note that we can split the nine rank-2 errors $J_\alpha J_\beta$ into six symmetric errors and three anti-symmetric errors \cite{sakurai}. The six symmetric errors are 
\begin{align}
    J_+ J_+ \qquad J_z & J_z \qquad J_- J_- \\
    J_+ J_- + J_- J_+ \quad J_0 J_+ +& J_+ J_0 \quad J_0 J_- + J_- J_0.
\end{align}
Each of these errors satisfies the KL conditions by \cref{lem:covarsymmetric}. 

On the other hand, the 3 anti-symmetric errors are $J_\alpha J_\beta - J_\beta J_\alpha$. The span of these errors is equal to the span of the rank-1 errors. Because we are assuming the spin code satisfies all rank-1 errors, then these rank-2 conditions will be satisfied as well. 

This concludes the proof of \cref{lem:rank2}.

\section{Proof of \cref{thm:auto1unique}}

\begin{proof}
One can compute $\widetilde{\Gamma}_{(\G,\lambda)}$ in \cref{eqn:calculation} for other finite subgroups $\G$ and faithful 2-dimensional irreps $\lambda$. 

There are three faithful 2-dimensional irreps for $\tet$, two faithful 2-dimensional irreps for $\oct$, and two faithful 2-dimensional irreps for $\ico$ (see the Supplemental Material of \cite{gross1}). None of these make $\widetilde{\Gamma}$ identically 0 for all $g' \in \G$ except $(\ico, \overline{\pi_2})$. The same is true for the faithful 2-dimensional irreps of the groups $ \mathrm{Dic}_n $.

Thus the automatic protection from $(\ico, \overline{\pi_2})$ is unique. 
\end{proof}

\section{Construction of $\ico$ codes}

Constructing a permutationally invariant $ ((n,2,3)) $ code transforming in $ (\G,\lambda)=(\ico, \overline{\pi_2}) $ is done in four steps. 

Step 1. Construct the $ (\G,\lambda) $ spin projector:
\[
    \Pi_{\G}= \frac{1}{60} \sum_{g \in \ico } \chi_{\overline{\pi_2}}(g)^* D^j(g). \numberthis
\]
Here $ 2j=n $. Note that for any $ g \in \SU(2) $ it is the case that $ \Y g \Y^{-1}= g^* $. Thus for any finite subgroup of $ \SU(2) $ containing $ \Y $ we have that $ g $ and $ g^* $ are in the same conjugacy class and thus take the same character value. Since the character $\chi_{\overline{\pi_2}}$ is real, and $ g $ is in the same conjugacy class as $ g^* $, the projector $ \Pi_{\G} $ is real and symmetric (since projectors are Hermitian). Note that this general fact was also noticed in \cite{gross2} for the special case of $\oct$. 

Step 2. Construct the projector onto the $ +1 $ eigenspace of logical $ Z $: 
\[
    \Pi_{\Z}= \frac{\indicator+iD^j(\Z)}{2}. \numberthis
\] 
Note that $ \Pi_{\Z} $ is also real and symmetric. 

Step 3. Since $ \Pi_\G $ and $ \Pi_{\Z} $ commute, the product $ \Pi_\G \Pi_\Z $ is symmetric and real as well. It follows that the eigenvectors of $ \Pi_\G \Pi_\Z $ are real. If we take \textit{any} real linear combination of the $ +1 $ eigenvectors of $ \Pi_\G \Pi_\Z $ as our $\ket{\overline{0}}$, and define $\ket{\overline{1}} $ by $ i D^j(\X) \ket{0}$, then we will have a $d=3$ spin code by \cref{lem:rank1} and \cref{lem:rank2}.

Step 4. Take the $ \ico $ spin $j$ code with distance $d = 3$ from step 3 and apply the Dicke state mapping $\D$ from \cref{eqn:dicke} to get a permutationally invariant $ \ico $-transversal $ ((n,2,3)) $ multiqubit code. 

\emph{Small Examples.---} Using this method we can construct $\ico$ multiqubit codes for the first few values of $n$ in \cref{thm:family}. The first few values of $j$ for which a $\overline{\pi_2}$ irrep of $\ico$ appears are $j = 7/2, 13/2, 17/2$ (corresponding to $n = 7, 13, 17$ respectively). We already saw the $n = 7$ case in \cref{code:us}. Now we list the codewords for the other two cases. 

A normalized basis for the $((13,2,3))$ code is
\begin{align*}
    \logzero &= \tfrac{1}{64} \bigg( 3\sqrt{55}\ket{D_0^{13}} + \sqrt{858} \ket{D_2^{13}} + \sqrt{13} \ket{D_4^{13}}\\
    &\quad\qquad  - 2\sqrt{39}\ket{D_6^{13}} - 5\sqrt{65}\ket{D_8^{13}} \\
    &\quad\qquad  + 3\sqrt{26}\ket{D_{10}^{13}} - \sqrt{715}\ket{D_{12}^{13}} \bigg) ,\\
    \logone &= X^{\otimes 13} \logzero.
\end{align*}

A normalized basis for the $((17,2,3))$ code is
\begin{align*}
    \logzero &= \tfrac{1}{192} \bigg( 3 \sqrt{390}\ket{D_0^{17}} -\sqrt{663} \ket{D_2^{17}} + \sqrt{9282} \ket{D_4^{17}} \\
    &\quad + \sqrt{357} \ket{D_6^{17}} + 4 \sqrt{561} \ket{D_8^{17}} - \sqrt{561} \ket{D_{10}^{17}} \\
    & \quad -\sqrt{3570} \ket{D_{12}^{17}} + \sqrt{3315} \ket{D_{14}^{17}} + \sqrt{6630} \ket{D_{16}^{17}} \bigg), \\
    \logone &= X^{\otimes 17} \logzero.
\end{align*}

\section{The $\ico$ family is larger than it appears}\label{sec:2infinity}

The family in \cref{thm:family} looks countably infinite, since it is indexed by $n$. However, within each $n$, there are uncountably many distinct codes. 

Let $\mu$ denote the multiplicity of the irrep $\overline{\pi_2}$ in spin $j = n/2$. Then there is a real projective space $\mathbb{R}P^{\mu-1}$ worth of $\ico$ error-correcting codes (the codes are distinct in that they are not equivalent by non-entangling gates) because of the results of \cref{lem:rank1}.

When $ \mu=1 $, there is a unique code up to code equivalence (for example the case of $ n=7,13,17 $ given above). But for any larger multiplicity $ \mu>1 $, there are infinitely many distinct $ \ico $ codes. The first value of $ n $ for which this is relevant is $ n=37 $, which has $ \mu=2 $ and thus there is an $ \mathbb{R}P^{1}= S^1  $ worth of inequivalent codes. In general, $ \mu \approx 1+ n/30$. So for larger numbers of qubits this method produces larger and larger manifolds worth of $ ((n,2,3)) $ codes with transversal gate group $ \ico $.

\section{Getting to Universality from $ \ico $}\label{sec:magic}

\emph{Super Golden Gate Magic.---} Having accomplished the ``cheap" part of the super golden gate implementation in this work, the next step in a practical implementation of the icosahedral super golden gate set would be a fault tolerant implementation of the $ \tau_{60} $ gate, perhaps using some icosahedral analog of magic state distillation. A final judgement on the efficiency of $\ico + \tau_{60}$ versus $\oct + T$ is impossible until this step is completed. 

\emph{Conventional Magic.---} Even without exotic magic state distillation, a code with transversal $ \ico $ is interesting in the sense that $ \ico + T$ is still a universal gate set. Although $ \ico + T$ is not a super golden gate set like $\ico + \tau_{60}$,  it still has very good approximation properties (and it's plausible that in some contexts it may outperform $\text{Clifford} + T$, since $\ico$ has more elements than $\oct$). $\ico + H$ and $\ico + S$ are also universal gate sets, and they are especially notable because, in both cases, the ``expensive" gate is from the Clifford group. This contrasts with the standard choice $\text{Clifford} + T$ where the expensive gate $ T $ is from the 3rd level of the Clifford hierarchy.

\emph{Code Switching.---} Lastly, a $\ico$ code might reduce the cost of code switching \cite{codeswitching}. Instead of $G + \tau$, consider two codes $\mathcal{C}_1$ and $\mathcal{C}_2$ that contain complementary transversal gates, i.e., the union of the transversal gates yields a universal gate set. The canonical example is the $[[7,1,3]]$ Steane code as $\mathcal{C}_1$, which supports $\oct$ transversally, and the $[[15,1,3]]$ Reed-Muller code as $\mathcal{C}_2$, which supports a transversal $T$ gate. Together these form an $n = 7 \cdot 15 = 105$ code with a single qubit fault tolerant universal gate set. However, careful arguments must be made regarding both distance and fault tolerance.

Theoretically, one could consider extending this argument to the $[[5,1,3]]$ code instead for $\mathcal{C}_1$, which supports $\tet$ transversally, because $\tet \cup \{T\}$ is still universal for $\U(2)$. Then one would have an $n = 5\cdot 15 = 75$ qubit code. But the $[[15,1,3]]$ code is the smallest code known that supports a transversal implementation of $T$ \cite{smallestT} so this seems to be the smallest $n$ one could achieve (since $[[5,1,3]]$ is the smallest code that has $d = 3$).

However, suppose we take the $((7,2,3))$ $\ico$ code as $\mathcal{C}_1$. As stated above, $\ico \cup \{H\}$ is universal for $\U(2)$, so we only need a code that has transversal Hadamard for $\mathcal{C}_2$. The smallest such error-correcting code is the $[[7,1,3]]$ Steane code. With these choices we would have $n = 7 \cdot  7 = 49$ which, at least naively, seems to be optimal in $n$. To be sure, careful analysis would need to be done in order to guarantee error correction properties and fault tolerance, but we leave this for future work.

\end{document}